\documentclass[11pt,a4paper]{article}
\pdfoutput=1
\usepackage[a4paper,margin=2.5cm]{geometry}
\usepackage[utf8]{inputenc}
\usepackage[T1]{fontenc}
\usepackage{lmodern}
\usepackage[bookmarksnumbered=true,linktocpage,hypertexnames=false,colorlinks=true,linkcolor=blue,urlcolor=blue,citecolor=blue,anchorcolor=green,bookmarks=true,pagebackref=true,breaklinks=true,pdfusetitle=true,pdfpagelabels=true,pdfstartview={FitH},pdfnewwindow=true]{hyperref}
\usepackage[dvipsnames]{xcolor}
\usepackage{bbm,physics,microtype,mathrsfs,amsmath,amssymb,amsthm,amsfonts,latexsym,mathtools,graphicx,enumitem,booktabs,bm,xspace,float,mathdots,caption,subcaption,ellipsis,enumitem,mleftright,tikz,comment,algorithm2e,slantsc}
\usepackage[T1]{fontenc}
\usepackage{thm-restate}

\usepackage[english]{babel}
\usepackage[capitalize,nameinlink]{cleveref}
\usepackage{authblk}
\newtheorem{theorem}{Theorem}[section]\crefname{theorem}{Theorem}{Theorems}
\newtheorem{lemma}[theorem]{Lemma}\crefname{lemma}{Lemma}{Lemmas}
\newtheorem{proposition}[theorem]{Proposition}\crefname{proposition}{Proposition}{Propositions}
\newtheorem{corollary}[theorem]{Corollary}\crefname{corollary}{Corollary}{Corollaries}
\crefname{fact}{Fact}{Facts}
\newtheorem{claim}[theorem]{Claim}\crefname{claim}{Claim}{Claims}
\crefname{example}{Example}{Examples}
\newtheorem{problem}[theorem]{Problem}\crefname{problem}{Problem}{Problems}
\theoremstyle{definition}
\newtheorem{definition}[theorem]{Definition}\crefname{definition}{Definition}{Definitions}
\theoremstyle{remark}
\crefname{remark}{Remark}{Remarks}
\theoremstyle{plain}
\newcommand{\thistheoremname}{}
\newtheorem{genericthm}[theorem]{\thistheoremname}

\DeclareMathOperator{\polylog}{polylog}
\DeclareMathOperator{\supp}{supp}
\DeclareMathOperator{\TV}{TV}

\newcommand{\eps}{\varepsilon}
\newcommand{\R}{\mathbbm R}

\newcommand{\SL}{\mathsf{SL}}
\renewcommand{\L}{\mathsf{L}}

\newcommand{\E}{\mathop{\bf E\/}}

\renewcommand{\Pr}{\mathop{\bf Pr\/}}
\numberwithin{equation}{section}
\allowdisplaybreaks[4]
\graphicspath{{./images/}}
\RestyleAlgo{boxruled}
\LinesNumbered

\newcommand{\eq}[1]{\hyperref[eq:#1]{(\ref*{eq:#1})}}
\renewcommand{\sec}[1]{\hyperref[sec:#1]{Section~\ref*{sec:#1}}}
\newcommand{\thm}[1]{\hyperref[thm:#1]{Theorem~\ref*{thm:#1}}}
\newcommand{\lem}[1]{\hyperref[lem:#1]{Lemma~\ref*{lem:#1}}}
\newcommand{\cor}[1]{\hyperref[cor:#1]{Corollary~\ref*{cor:#1}}}
\newcommand{\app}[1]{\hyperref[app:#1]{Appendix~\ref*{app:#1}}}
\newcommand{\tabl}[1]{\hyperref[tab:#1]{Table~\ref*{tab:#1}}}
\newcommand{\defin}[1]{\hyperref[def:#1]{Definition~\ref*{def:#1}}}
\newcommand{\fig}[1]{\hyperref[fig:#1]{Figure~\ref*{fig:#1}}}
\newcommand{\clm}[1]{\hyperref[clm:#1]{Claim~\ref*{clm:#1}}}
\newcommand{\conj}[1]{\hyperref[conj:#1]{Conjecture~\ref*{conj:#1}}}
\newcommand{\rem}[1]{\hyperref[rem:#1]{Remark~\ref*{rem:#1}}}
\newcommand{\probl}[1]{\hyperref[prob:#1]{Problem~\ref*{prob:#1}}}
\newcommand{\thmthm}[2]{\hyperref[thm:#1]{Theorem~\ref*{thm:#1}} and~\hyperref[thm:#2]{\ref*{thm:#2}}}
\newcommand{\lemlem}[2]{\hyperref[lem:#1]{Lemma~\ref*{lem:#1}} and~\hyperref[lem:#2]{\ref*{lem:#2}}}

\begin{document}

%=============================================================================
\title{(No) Quantum space-time tradeoff for USTCON}
\author[1]{Simon Apers}
\author[2]{Stacey Jeffery\thanks{Supported by ERC STG grant 101040624-ASC-Q, NWO Klein project number OCENW.Klein.061, and ARO contract no W911NF2010327. SJ is a CIFAR Fellow in the Quantum Information Science Program.}}
\author[3,4]{Galina Pass\thanks{Supported by the National Agenda for Quantum Technologies (NAQT), as part of the Quantum Delta NL programme.}}
\author[4]{Michael Walter\thanks{Supported by the European Research Council~(ERC) through ERC Starting Grant 101040907-SYMOPTIC, the Deutsche Forschungsgemeinschaft (DFG, German Research Foundation) under Germany's Excellence Strategy - EXC\ 2092\ CASA - 390781972, the Federal Ministry of Education and Research (BMBF) through project Quantum Methods and Benchmarks for Resource Allocation (QuBRA), and NWO grant OCENW.KLEIN.267.}}
\affil[1]{CNRS, IRIF, Paris}
\affil[2]{CWI \& QuSoft}
\affil[3]{Korteweg-de Vries Institute for Mathematics \& QuSoft, University of Amsterdam}
\affil[4]{Faculty of Computer Science, Ruhr University Bochum}
\date{}
\maketitle
\begin{abstract}
Undirected $st$-connectivity is important both for its applications in network problems, and for its theoretical connections with logspace complexity. Classically, a long line of work led to a time-space tradeoff of $T=\widetilde{O}(n^2/S)$ for any $S$ such that $S=\Omega(\log (n))$ and $S=O(n^2/m)$.
Surprisingly, we show that quantumly there is no nontrivial time-space tradeoff: there is a quantum algorithm that achieves both optimal time $\widetilde{O}(n)$ and space $O(\log (n))$ simultaneously.
This improves on previous results, which required either $O(\log (n))$ space and $\widetilde{O}(n^{1.5})$ time, or $\widetilde{O}(n)$ space and time.
To complement this, we show that there is a nontrivial time-space tradeoff when given a lower bound on the spectral gap of a corresponding random walk.
\end{abstract}

%=============================================================================
\section{Introduction}
%=============================================================================
For an undirected graph $G=(X,E)$ on $n=\abs{X}$ vertices and $m=\abs{E}$ edges, with $s,t\in X$, $st$-connectivity or \textsc{ustcon} is the problem of deciding whether~$s$ and~$t$ are in the same component.
This problem has applications in many other graph and network problems, and is of theoretical importance for its connection with space complexity (see e.g.~\cite{wigderson1992complexity}).
In particular, \textsc{ustcon} is complete for the class \emph{symmetric logspace}, $\SL$, which was shown to be equal to \emph{logspace}, $\L$, by exhibiting a classical deterministic logspace algorithm for \textsc{ustcon}~\cite{reingold2008SL}.
In this paper, we consider quantum algorithms for this problem.

There are different versions of the problem \textsc{ustcon} depending on how $G$ is accessed.
If~$G$ is given as an adjacency matrix, we denote the problem \textsc{ustcon}$_{\text{mat}}$.
If~$G$ is given as an array of arrays, one for each vertex, enumerating the neighbours, we denote the problem \textsc{ustcon}$_{\text{arr}}$.
\footnote{There are variations on the details of this model.
For now, we allow \textsc{ustcon}$_{\text{arr}}$ to stand in for multiple variations of the array access model, but precise details of the variations can be found in \sec{ustcon}.}
If one only cares about space complexity, these problems are equivalent, but the same is not true of time complexity: adjacency queries can simulate an array query, and vice versa, in logspace, but there is a non-negligible time overhead.

A classical deterministic algorithm based on breadth-first search or depth-first search can solve \textsc{ustcon}$_{\text{arr}}$ in $\widetilde{O}(m)$ time, using $\widetilde{O}(n)$ space.
Using a random walk, the space complexity can be improved to $O(\log (n))$, at the expense of $\widetilde{O}(nm)$ time complexity~\cite{aleliunas1979ustconn}.
A series of works~\cite{broder1989trading,beame1990time,edmonds1993time,barnes1993short,feige1993randomized} culminated in a space-time tradeoff for \textsc{ustcon}$_{\text{arr}}$ of $T=\widetilde{O}(n^2/S)$ queries for any space bound $S = \Omega(\log (n))$ and $S = O(n^2/m)$, due to Kosowski~\cite{MH}.
While there is no matching time-space lower bound, it is unlikely that this tradeoff can be significantly improved (see \cite[Section 5.1 of arXiv v2]{MH} for a discussion).
Kosowski's algorithm is based on using Metropolis-Hastings random walks to find connections between $S$ sampled vertices and $s,t$ until it is becomes possible to conclude that $s$ and $t$ are connected.
For comparison, in the adjacency matrix model, the randomized query complexity of \textsc{ustcon}$_{\text{mat}}$ is $\widetilde{\Theta}(n^2)$ and there is no space-time tradeoff.

A quantum algorithm of D\"urr, Heiligman, H{\o}yer and Mhalla~\cite{DHHM} for \textsc{connectivity} can be adapted to solve \textsc{ustcon}$_{\text{mat}}$ in $\widetilde{O}(n^{1.5})$ time and \textsc{ustcon}$_{\text{arr}}$ in $\widetilde{O}(n)$ time, both of which are optimal up to polylog factors.
Both of these algorithms use $\widetilde{O}(n)$ space, of which all but  $O(\log (n))$ can be classical space (assuming quantum RAM access).
A subsequent quantum algorithm for \textsc{ustcon}$_{\text{mat}}$ due to Belovs and Reichardt uses $\widetilde{O}(n^{1.5})$ time, but only $O(\log (n))$ space~\cite{BR12}, which is optimal in terms of both space and time.  It is also possible to solve \textsc{ustcon}$_{\text{arr}}$ in $O(\log (n))$ space and $\widetilde{O}(\sqrt{nm})$ time, using a quantum walk (see for example~\cite{belovs2013electric}). This quantum walk algorithm requires a quantum version of array access to the input graph, which we refer to as \textsc{ustcon}$_{\text{qw}}$ in the next section.

\begin{table}
    \centering
    \renewcommand*{\arraystretch}{1.5}
    \begin{tabular}{|r|c|c|}
    \hline
     \multicolumn{3}{|c|}{\textsc{ustcon}$_{\text{mat}}$}\\
    \hline
     & Time & TS-tradeoffs \\
    \hline
     Classical & $\widetilde{\Theta}(n^2)$ & $S=O(\log (n))$, $T=\widetilde{O}(n^3/d)$ \\
      & & $S=\widetilde{O}(n)$, $T=\widetilde{O}(n^2)$\\
     \hline
     Quantum & $\widetilde{\Theta}(n^{1.5})$ & $S=O(\log (n))$, $T=\widetilde{O}(n^{1.5})$~\cite{BR12}\\
     \hline
     \multicolumn{3}{|c|}{\textsc{ustcon}$_{\text{arr}}$}\\
    \hline
     & Time & TS-tradeoffs \\
    \hline
    Classical & $\widetilde{\Theta}(m)$ & $T=\widetilde{O}(\max\{n^2/S,m\})$ \\
    \hline
    Quantum  & $\widetilde{\Theta}(n)$ & $S=O(\log (n))$, $T=\widetilde{O}(n^{1.5})$ \\
     & & $S=T=\widetilde{O}(n)$~\cite{DHHM}\\
    \hline
    & & \textbf{This work:} $S=O(\log (n))$, $T=\widetilde{O}(n)$ \\
    \hline
    \end{tabular}
    \caption{A summary of classical (randomized) and quantum time and space complexities for \textsc{ustcon} in the adjacency matrix and adjacency array models. The classical results for \textsc{ustcon}$_\text{mat}$ follow from (1) the $\log (n)$-space result for \textsc{ustcon}$_{\text{arr}}$ with an $n/d$ overhead for finding neighbours of the current vertex in a $d$-regular graph; and (2) BFS.}
    \label{tab:my_label}
\end{table}

%-----------------------------------------------------------------------------
\subsection{Summary of results}
%-----------------------------------------------------------------------------
We describe new quantum walk algorithms for \textsc{ustcon}$_{\text{arr}}$.
These algorithms consider a quantum walk version of the adjacency array model, in which the input graph is accessed by a quantum analogue of classical random walk steps.
Recall that in the adjacency array model, we assume that for any vertex $u$, we can query, for any $i\in [d_u]$, the $i$-th neighbour of $u$, $v_i(u)$.
Then a random walk step can be performed from state $u$ by sampling a uniform $i\in [d_u]$, and then computing~$v_i(u)$, which becomes the current state.
In the \emph{quantum walk access model}, we assume that for any vertex $u$, we can prepare a uniform superposition over the neighbours of~$u$.
While these models are not identical, they are very similar, and in \sec{ustcon}, we formally define the models, and show that quantum walk access can be simulated in the array model with polylogarithmic overhead under reasonable additional assumptions.

Letting \textsc{ustcon}$_{\text{qw}}$ denote the $st$-connectivity problem in the quantum walk access model,
we present a one-sided error quantum algorithm that solves \textsc{ustcon}$_{\text{qw}}$ in time~$\widetilde{O}(n)$ and space~$O(\log (n))$.
Perhaps surprisingly, this means that \textsc{ustcon}$_{\text{qw}}$ admits \emph{no} nontrivial tradeoff between space and time in the quantum setting -- a single algorithm can solve this problem optimally in terms of both time and space (see \thm{single-walk} for the formal result).

\begin{theorem}[Informal]\label{thm:intro1}
There is a $O(\log (n))$-space quantum algorithm that decides $\textsc{ustcon}_{\mathrm{qw}}$ with one-sided error in $\widetilde{O}(n)$ time.
\end{theorem}

In this paper, when we say \emph{time}, we are counting: (1) quantum gates (unitaries that act on at most a constant number of qubits); (2) quantum walk queries to $G$; and (3) (quantum) random access (QCRAM) operations (QCRAM is used in our second algorithm only, see below).
Inspired by~\cite{MH}, our algorithm is based on a quantum walk search for~$t$ starting from~$s$ using a random walk that can be interpreted as a Metropolis-Hastings random walk.

Because of the close relationship between \textsc{ustcon} and classical logspace, we can consider what this means for logspace problems in general.
It does not mean that more space does not reduce the quantum time complexity of \emph{any} problem, but it is interesting to consider:
in what settings do we get a non-trivial time-space tradeoff?
We consider one such setting: when we are given a promise on the spectral gap or mixing time of the random walk on $G$ (see \sec{random-walks}).
In that case, we prove the following theorem (see \thm{trade-off} for the formal result).

\begin{theorem}[Informal]\label{thm:intro2}

Suppose whenever $s$ and $t$ are connected, the random walk spectral gap is at least $\delta > 0$.
For any $S \in \Omega(\log (n))$, there is a quantum algorithm that decides $\textsc{ustcon}_{\mathrm{qw}}$ with bounded error in $O(S)$ space and $T=\widetilde{O}\left(\frac{S}{\delta}+\sqrt{\frac{n}{\delta S}}\right)$ time.
\end{theorem}

\noindent
Our algorithm adapts~\cite{apers2019qsampling} in a way that allows exploiting the spectral gap promise.
The time bound decreases monotonically for $S \in \Omega(\log n)$ until $S \in O((n\delta)^{1/3})$, at which point it reaches time complexity $T = \widetilde O(n^{1/3}/\delta^{2/3})$.
We leave it as an open problem to prove a matching lower bound (at least for some values of $\delta$), which would prove that in certain regimes, it is not possible to achieve optimal time and space simultaneously.

In the space bound $S$ of both algorithms in \cref{thm:intro1,thm:intro2}, only $O(\log (n))$ memory needs to be actual quantum workspace (i.e., qubits).
The remaining $O(S)$ memory can be classical RAM in the first algorithm and QCRAM in the second algorithm, that is, classical RAM that is queryable at a quantum superposition of addresses.
We discuss the latter in \cref{subsec:qram}.

We summarize our results in \cref{fig:diagram}. For $S=\log(n)$, the algorithm of \thm{intro2} has a worse time complexity than the algorithm of \thm{intro1}, whenever $\delta<\frac{1}{n}$. We leave it as an open problem to give a single algorithm that is optimal for all $\delta$.

\begin{figure}[htb]
\centering
\includegraphics[width=.65\textwidth]{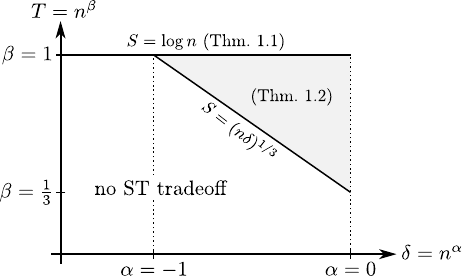}
\caption{Quantum space-time tradeoffs for USTCON, with axes representing the time complexity and spectral gap promise (up to polylog-factors). The grey area represents the regime in which a non-trivial tradeoff is achieved. \cref{thm:intro1} (upper line) corresponds to the regime with space $S = O(\log n)$ and time $T = \widetilde O(n)$. \cref{thm:intro2} (grey area) corresponds to the regime with a promise on $\delta$, and interpolates between $S = O(\log n)$ and $T = \widetilde O(n)$, and $S = O((n\delta)^{1/3})$ and $T = \widetilde O(n^{1/3}/\delta^{2/3})$.}
\label{fig:diagram}
\end{figure}

\paragraph{Organization:} The remainder of this paper is organized as follows. We describe preliminaries in \sec{prelim} and \sec{ustcon}. In \sec{single-walk}, we prove \thm{intro1} by exhibiting a quantum algorithm for \textsc{ustcon}$_{\text{qw}}$ that is optimal in both time and space. For completeness, we also include a proof of a corresponding lower bound in \sec{lower-bound}. In \sec{tradeoff}, we prove \thm{intro2} exhibiting a quantum time-space tradeoff when given a promise on the spectral gap.

%=============================================================================
\section{Preliminaries}\label{sec:prelim}
%=============================================================================
We first give some general notation.
For a positive integer $k$, we let $[k]=\{1,\dots,k\}$.
Throughout this work, $n$ denotes the number of vertices and $m$ the number of edges of the input graph.
For any function $f$, we let $\widetilde{O}(f(n)) = f(n)\cdot{\polylog}(n)$.

%-----------------------------------------------------------------------------
\subsection{Probability theory}\label{sec:prob}
%-----------------------------------------------------------------------------

A \emph{(probability) distribution} on a finite set~$X$ is a non-negative function~$\sigma\colon X\to\R_{\geq0}$ such that $\sum_{v \in X} \sigma(v) = 1$.
Its \emph{support} is defined as $\supp(\sigma) := \{ v \in X : \sigma(v) > 0 \}$.
We will implicitly identify such~$\sigma$ with \emph{row} vectors, as is customary in the random walk literature.
To any distribution $\sigma$, we also associate a quantum state~$\ket\sigma:=\sum_{v\in X}\sqrt{\sigma(v)}\ket{v}$.
Measuring $\ket{\sigma}$ in the standard basis returns a sample from $\sigma$.

For any distribution $\sigma$ on $X$, and any subset $M\subseteq X$, we will let $\sigma(M)=\sum_{u\in M}\sigma(u)$.
We let~$\sigma_M$ denote the \emph{normalized restriction} of $\sigma$ to $M$, defined by $\sigma_M(u) = \sigma(u)/\sigma(M)$ for all~$u\in M$ and $\sigma_M(u) = 0$ elsewhere.

Finally, the \emph{total variation distance} between two distributions $\sigma$ and $\tau$ on $X$ is defined as
\[
  \norm{\sigma-\tau}_{\TV}
:= \frac{1}{2}\sum_{u\in X} \abs{ \sigma(u)-\tau(u) }
= \max_{A \subseteq X} \abs{ \sigma(A) - \tau(A) }.
\]

%-----------------------------------------------------------------------------
\subsection{Random walks}\label{sec:random-walks}
%-----------------------------------------------------------------------------

Fix an undirected graph $G=(X,E)$ with $n=\abs{X}$ vertices and $m=\abs{E}$ edges.
We take $E \subseteq \binom X 2$, that is, edges~$e \in E$ are subsets~$e = \{u,v\} = \{v,u\}$ of pairs of vertices.
We will let
\[ N(u):=\{v\in X:\{u,v\}\in E\} \]
denote the neighbourhood of $u\in X$, and $d_u = \abs{N(u)}$ the degree of $u$.
For convenience we assume that all vertices have positive degree.

Fix edge weights given by a symmetric matrix $W \in \R_{\geq0}^{X \times X}$ such that $W_{u,v}=W_{v,u}$ for all $u,v\in X$, and~$W_{u,v} > 0$ if and only if $\{u,v\} \in E$. Then $G=(X,E,W)$ defines a \emph{weighted graph}. When no $W$ is given, the graph is \emph{unweighted} and we let $W_{u,v}=1$ for all $\{u,v\}\in E$.
For~$u \in X$, define $w_u = \sum_{v \in X} W_{u,v}$.
The corresponding \emph{(weighted) random walk} is the reversible Markov chain on $X$ with \emph{transition matrix}~$P \in \R_{\geq0}^{X \times X}$ given by
\begin{align}\label{eq:transmat}
    P_{u,v} = \begin{cases}
      \frac{W_{u,v}}{w_u} & \text{if $\{u,v\} \in E$} \\
      0 & \text{otherwise}
    \end{cases}
      \qquad \forall u, v \in X.
\end{align}
This means that the probability of moving from the vertex~$u$ along an edge to a neighbouring vertex~$v$ is proportional to the edge's weight.
In the unweighted case, this is called the \emph{simple random walk}; in each step it simply moves to a neighbouring vertex chosen uniformly at random.

Let $\pi\in\R^X_{>0}$ be the distribution defined by
\begin{align*}
    \pi(u) = \frac{w_u}{{\cal W}(G)} \qquad \forall u \in X,
\end{align*}
where ${\cal W}(G) = \sum_{u \in X} w_u = \sum_{u,v \in X} W_{u,v}$.
In the unweighted case, $\pi$ is proportional to the degree.
The distribution~$\pi$ is a \emph{stationary} distribution of the random walk, i.e., $\pi P = \pi$ (it is a left eigenvector of~$P$ with eigenvalue~$1$).

In fact, when the graph~$G$ is connected, $\pi$ is also the \emph{unique} stationary distribution of~$P$.
If in addition the graph is not bipartite, then all other eigenvalues have absolute value strictly less than one.
That is, if $1 = \lambda_1 \geq \dots \geq \lambda_n \geq -1$ are the eigenvalues of~$P$ then the \emph{(absolute) spectral gap} $\gamma_\star=\gamma_\star(G) := \min \{ 1 - \abs{\lambda_j} : j = 2,\dots,n \} = \min \{ 1 - \lambda_2, 1 + \lambda_n \}$ is strictly positive.
Importantly, the inverse of the spectral gap bounds the random walk's \emph{mixing time}, that is, the time required for convergence to the stationary distribution:

\begin{theorem}[{\cite[Thm.~12.4]{levin2017markov}}]\label{thm:mix}
Assume $G$ is connected and not bipartite.
Let $\eps>0$ and
\begin{align*}
  t \geq \frac1{\gamma_\star} \log \left(\frac 1 {\eps \pi_{\min}}\right),
\end{align*}
where $\pi_{\min} = \min_{u\in X} \pi(x)$.
Then $\norm{ \sigma P^t - \pi }_{\TV} \leq \eps$ for any distribution $\sigma$ on $X$.
\end{theorem}

Conversely, it is known that $t \geq ( \frac1{\gamma_\star} -  1 ) \log(\frac1{2\eps})$ is necessary to ensure mixing from an arbitrary initial distribution~\cite[Thm.~12.5]{levin2017markov}.
In the unweighted case, we have $\pi_{\min} \geq \frac{d_{\min}}{n \, d_{\max}} \geq \frac1{n^2}$, so the former shows that \cref{thm:mix} is tight up to $\log(n)$ factors in that case.

Finally, for any $s,t\in X$ we let ${\cal H}_{s,t}$ denote the \emph{hitting time} from $s$ to $t$, which is the expected number of steps needed to reach $t$ in a random walk starting from $s$.
We let ${\cal C}_{s,t}={\cal H}_{s,t}+{\cal H}_{t,s}$ denote the \emph{commute time} between $s$ and $t$ -- the expected number of steps needed to reach $t$ and then return to $s$ in a random walk starting from $s$. These quantities are finite if and only if~$s$ and~$t$ are in the same component of $G$. More generally, the commute time ${\cal C}_{s,M}$ from $s$ to a subset $M \subseteq X$ is the expected number of steps needed to reach any vertex in $M$ and then return to $s$ in a random walk starting from $s$.

%-----------------------------------------------------------------------------
\subsection{Quantum walk search algorithms}\label{subsec:qwalksearch}
%-----------------------------------------------------------------------------

Quantum walk search refers to the use of quantum walks to find certain ``marked'' elements on a graph.
We will use quantum walk search to search for a vertex connected to $t$ in the connected component of $S$.
Specifically, we will use the following special case of \cite[Thm.~13]{QWS}.\footnote{To see that this follows from \cite[Thm.~13]{QWS}, note that when $\ket{\sigma}=\ket{s}$, the cost to set up $\ket{\sigma}$ is $\log (n)$ and the value $C_{\sigma,M}$ from \cite{QWS} is exactly the commute time from $s$ to $M$~\cite[Thm.~4]{QWS}.}

\begin{theorem}\label{thm:quantum_walk_search}
Let $P$ be a random walk on a weighted graph with vertex set $X$, $M \subseteq X$ a subset of ``marked'' vertices, and $s\in X$. Let ${\sf C}$ be the (quantum) time complexity to check for a given~$u\in X$ whether~$u\in M$, let ${\sf U}$ be the time complexity of implementing the weighted quantum walk oracle
\begin{equation*}
\ket{u}\ket{0}\mapsto \sum_{v\in N(u)}\sqrt{P_{u,v}}\ket{u}\ket{v}.
\end{equation*}
in space $O(\log(n))$.
Let ${\cal C}$ be a known upper bound on the commute time ${\cal C}_{s,M}$ in the case where $s$ and $M$ are connected (and in particular $M\neq\emptyset$).
Then there is a quantum algorithm that, if $M \neq \emptyset$ and $s$ is connected to $M$, finds an element of $M$ with probability at least $2/3$.
If $M = \emptyset$ or $s$ is not connected to $M$, then the algorithm outputs a vertex not in $M$.
The algorithm has time complexity ${O}(\sqrt{\log(\mathcal C)}\log(n)+\sqrt{{\cal C}\log (\mathcal C) \log(\log(\mathcal C))}({\sf C}+{\sf U}))$ and space complexity $O(\log(n))$.
\end{theorem}

%-----------------------------------------------------------------------------
\subsection{Quantum RAM}\label{subsec:qram}
%-----------------------------------------------------------------------------

Our algorithm will exploit the given space by saving sets of vertices which will be either connected to $s$ or to $t$.
For our quantum algorithm to access this space, we assume access to a so-called quantum-classical random access memory or \emph{QCRAM}.
This refers to a memory that only stores classical information, but can be queried at a superposition of addresses.
More specifically, an $R$-bit QCRAM stores a string of bits $q \in \{0,1\}^R$ so that the following operations are supported in time $\polylog(R)$:
\begin{enumerate}
\item
\emph{UPDATE(i, x):} store $x \in \{0,1\}$ in the $i$-th bit (i.e., set $q_i = x$).
\item
\emph{QUERY:} for any superposition $\sum_i \alpha_i \ket{i} \ket{s_i}$, it maps
\[
\sum_i \alpha_i \ket{i} \ket{s_i}
\mapsto \sum_i \alpha_i \ket{i} \ket{s_i \oplus q_i}.
\]
\end{enumerate}

As was first described by Kerenidis and Prakash \cite{kerenidis2016quantum}, using such a QCRAM we can set up a data structure to generate quantum superpositions over elements in the QCRAM.
We will use the following formulation based on \cite{apers2019qsampling}.
\begin{lemma}\label{lem:QCRAM}
Fix integer parameters $\ell$ and $k$.
Using an ${O}(k \ell \log(\ell))$-bit QCRAM, there is a data structure, $D$, that stores up to $\ell$ elements $x\in \{0,1\}^k$ with associated integer weights,~$c_x$, of bounded absolute value for some poly$(\ell)$ bound, and supports the following operations in time~$O(k\cdot\mathrm{polylog}(k\ell))$ per operation:
\begin{enumerate}
\item
insertion or deletion of a pair $(x,c_x)$,
\item
quantum queries of the form ``Is $x \in D$?'',
\item
preparation of the quantum state
$\frac{1}{\sqrt{\sum_{x \in D} c_x}} \sum_{x \in D} \sqrt{c_x} \ket{x}$.
\end{enumerate}
\end{lemma}

%-----------------------------------------------------------------------------
\section{USTCON and the Quantum Walk model}\label{sec:ustcon}
%-----------------------------------------------------------------------------
In this section we define the undirected $st$-connectivity problem (\textsc{ustcon}). The input to this problem is an undirected graph $G=(X,E)$. Classically, there are various ways this input may be given, which may change the complexity of the problem. For example, in the \emph{adjacency array model} (defined below), it is possible to randomly sample a neighbour of any vertex $u$ in $O(1)$ queries to $G$ (assuming access to the vertex degrees), facilitating a random walk on $G$, whereas if $G$ is given as an \emph{adjacency matrix}, a random walk step is not so simple.

We will work in a quantum walk analogue of the adjacency array model. We assume that $G$ can be accessed via the \textit{quantum walk oracle} that for every $u\in X$ outputs a uniform superposition over its neighbours:\footnote{Note that this is exactly the quantum walk oracle defined in \thm{quantum_walk_search}, specialized to unweighted graphs.}
\begin{equation}\label{eq:simple oracle}
  \mathcal{O}_W:\ket{u}\ket{0}\mapsto \frac{1}{\sqrt{d_u}}\sum_{v\in N(u)}\ket{u}\ket{v}.
\end{equation}

\noindent Formally, we describe \textsc{ustcon}$_{\text{qw}}$ in terms of the input and output.

\begin{problem}[$\textsc{ustcon}_{\text{qw}}$]\label{prob:ustcon}
Given access to an undirected graph $G=(X,E)$ via the quantum walk oracle $\mathcal{O}_W$, and two vertices $s,t\in X$, decide whether $s$ and $t$ are in the same connected component of $G$.
\end{problem}

To compare our work with classical results on \textsc{ustcon}$_{\text{arr}}$, we describe an implementation of the quantum walk oracle defined above based on adjacency array access to a graph.
Let $u\in X$ and $i\in [d_u]$. We assume that for each vertex $u$ there is a fixed numbering of its neighbours from~$1$ to~$d_u$. In the \textit{adjacency array model}, two types of queries are allowed:

\begin{itemize}
    \item Degree query $\mathcal{O}_D:\ket{0}\ket{u}\mapsto \ket{d_u}\ket{u}$
    \item Neighbour query $\mathcal{O}_N:\ket{u}\ket{i}\ket{0}\mapsto \ket{u}\ket{i}\ket{v_{i}(u)}$
\end{itemize}

\noindent In the \emph{sorted adjacency array model} we additionally assume that for every vertex $u\in X$ its neighbours are sorted: for any $i,j\in [d_u]$, if $i<j$ then $v_i(u)<v_j(u)$.
In particular, this allows us to check with $O(\log( n))$ queries whether a given pair of vertices $u,v$ are adjacent.
We define the $\textsc{ustcon}$-problem in this model as follows.

\begin{problem}[\textsc{ustcon}\textsubscript{s-arr}]\label{prob:ustcon-arr}
Given access to an undirected graph $G=(X,E)$ via the sorted adjacency array model and two vertices $s,t\in X$, decide whether $s$ and $t$ are in the same connected component of $G$.
\end{problem}

The question that we consider is how many sorted adjacency array queries to the graph it takes to implement the quantum walk oracle $\mathcal{O}_W$.

\begin{lemma} \label{lem:QW_assumption_implementation}
The quantum walk oracle $\mathcal{O}_W$ for an unweighted graph $G$ (\cref{eq:simple oracle}) with maximum degree $d_{\max}$ can be implemented with $O(\log (d_{\max}))$ queries in the sorted adjacency array model, and
$\widetilde O(1)$ other elementary operations and space.
\end{lemma}

\begin{proof}
Assume first, that there is an additional type of query allowed, namely:
\begin{equation}
    \text{Index query: } \mathcal{O}_I:\ket{u}\ket{v}\ket{0}\mapsto \ket{u}\ket{v}\ket{i}
\end{equation}
for $u,v\in X$ and $i\in [d_u]$ such that $v_i(u)=v$.
Then $\mathcal{O}_W$ can be implemented using ${\cal O}_I$, ${\cal O}_D$, and ${\cal O}_N$  as follows. Let $F_d$ denote the Fourier transform over $\mathbb{Z}_d$, and let $F=\sum_{d=1}^n\ket{d}\bra{d}\otimes F_d$, which can be implemented (to any inverse polynomial precision) in $O(\log (n))$ gates. Then for any $u\in X$, we implement:
\begin{align*}
\ket{0}\ket{u}\ket{0}\ket{0}&\overset{\mathcal{O}_D}{\mapsto}\ket{d_u}\ket{u}\ket{0}\ket{0}\overset{F}{\mapsto}\frac{1}{\sqrt{d_u}}\ket{d_u}\sum_{i=1}^{d_u}\ket{u}\ket{i}\ket{0}\\
&\overset{\mathcal{O}_N}{\mapsto}\frac{1}{\sqrt{d_u}}\ket{d_u}\sum_{i=1}^{d_u}\ket{u}\ket{i}\ket{v_i(u)}\overset{\mathcal{O}_I^{\dagger}\mathcal{O}_D^{\dagger}}{\mapsto}\frac{1}{\sqrt{d_u}}\sum_{i=1}^{d_u}\ket{u}\ket{v_i(u)} = {\cal O}_W\ket{u}\ket{0}.
\end{align*}
To complete the proof, note that the index query operator $\mathcal{O}_I$ only requires $O(\log (d_u))$ sorted adjacency array queries, since the neighbours are sorted and this makes it possible to perform binary search for $i$ such that $v_i(u)=v$.
\end{proof}

It follows that any quantum algorithm solving $\textsc{ustcon}_{\mathrm{qw}}$ in $T$ time and $S=\Omega(\log(n))$ space can solve $\textsc{ustcon}_{\text{s-arr}}$ in $\widetilde{O}(T)$ time and $O(S)$ space.

%=============================================================================
\section{Time- and space-optimal quantum algorithm}\label{sec:single-walk}
%=============================================================================

In \sec{MH_description} and \sec{no-tradeoff}, we give an algorithm for \textsc{ustcon}$_{\text{qw}}$ that is optimal in both time and space. For completeness, we first give a time lower bound in \sec{lower-bound}.

%-----------------------------------------------------------------------------
\subsection{Lower bound}\label{sec:lower-bound}
%-----------------------------------------------------------------------------
The proof of the following lower bound follows the lines of the proof of an analogous lower bound for the strong connectivity problem described in \cite{DHHM}. The proof is via a reduction from \textsc{parity}.

\begin{problem}[\textsc{parity}]
Given oracle access to a string $x\in \{0,1\}^n$ via $\mathcal{O}_x:\ket{i}\ket{b}\mapsto\ket{i}\ket{b\oplus x_i}$, return $\bigoplus_{i=0}^{n-1}x_i$.
\end{problem}

\begin{lemma}[\cite{beals2001QLowerBoundPoly,farhi1998parity}]\label{lem:parity_complexity}
The bounded error quantum query complexity of \textsc{parity} is~$\Omega(n)$.
\end{lemma}

\noindent We use \lem{parity_complexity} and a reduction from parity to show the following.

\begin{theorem}
The bounded error quantum query complexity of \textsc{ustcon}$_\mathrm{s\mbox{-}arr}$ and \textsc{ustcon}$_{\mathrm{qw}}$ is~$\Omega(n)$.

\end{theorem}

\begin{figure}
    \centering
    \begin{tikzpicture}
    \node at (0,1.25) {$s$};
    \filldraw (0,1) circle (.05);
    \filldraw (0,0) circle (.05);
    \node at (0,-.25) {$v_0'$};

        \draw[red,->] (0,0)--(1,1);
        \draw[red,->] (0,1)--(1,0);
        \node at (.07,.7) {\small\color{red}${x}_0$};
        \node at (.07,.3) {\small\color{red}$x_0$};

        \draw[blue,->] (0,0)--(1,0);
        \draw[blue,->] (0,1)--(1,1);
        \node at (.5,1.15) {\small\color{blue}$\bar{x}_0$};
        \node at (.5,-.17) {\small\color{blue}$\bar{x}_0$};

    \node at (1,1.25) {$v_1$};
    \filldraw (1,1) circle (.05);
    \filldraw (1,0) circle (.05);
    \node at (1,-.25) {$v_1'$};

        \draw[red,->] (1,0)--(2,1);
        \draw[red,->] (1,1)--(2,0);
        \node at (1.07,.7) {\small\color{red}${x}_1$};
        \node at (1.07,.3) {\small\color{red}$x_1$};

        \draw[blue,->] (1,0)--(2,0);
        \draw[blue,->] (1,1)--(2,1);
        \node at (1.5,1.15) {\small\color{blue}$\bar{x}_1$};
        \node at (1.5,-.17) {\small\color{blue}$\bar{x}_1$};

    \node at (2,1.25) {$v_2$};
    \filldraw (2,1) circle (.05);
    \filldraw (2,0) circle (.05);
    \node at (2,-.25) {$v_2'$};

        \draw[red,->] (2,0)--(3,1);
        \draw[red,->] (2,1)--(3,0);
        \node at (2.07,.7) {\small\color{red}${x}_2$};
        \node at (2.07,.3) {\small\color{red}$x_2$};

        \draw[blue,->] (2,0)--(3,0);
        \draw[blue,->] (2,1)--(3,1);
        \node at (2.5,1.15) {\small\color{blue}$\bar{x}_2$};
        \node at (2.5,-.17) {\small\color{blue}$\bar{x}_2$};

        \node at (3.5,.5) {$\dots$};

        \draw[red,->] (4,0)--(5,1);
        \draw[red,->] (4,1)--(5,0);

        \draw[blue,->] (4,0)--(5,0);
        \draw[blue,->] (4,1)--(5,1);
        \node at (4.5,1.15) {\small\color{blue}$\bar{x}_{n-2}$};
        \node at (4.5,-.17) {\small\color{blue}$\bar{x}_{n-2}$};

    \filldraw (5,1) circle (.05);
    \filldraw (5,0) circle (.05);

        \draw[red,->] (5,0)--(6,1);
        \draw[red,->] (5,1)--(6,0);
        \node at (6.2,.7) {\small\color{red}${x}_{n-1}$};
        \node at (6.2,.3) {\small\color{red}$x_{n-1}$};

        \draw[blue,->] (5,0)--(6,0);
        \draw[blue,->] (5,1)--(6,1);
        \node at (5.5,1.15) {\small\color{blue}$\bar{x}_{n-1}$};
        \node at (5.5,-.17) {\small\color{blue}$\bar{x}_{n-1}$};

    \node at (6.2,1.25) {$v_{n}$};
    \filldraw (6,1) circle (.05);
    \filldraw (6,0) circle (.05);
    \node at (6.2,-.25) {$t$};

    \end{tikzpicture}
    \caption{The parity graph. We include an edge labelled by ``$x_i$'' (in red) if and only if $x_i=1$, and an edge labelled ``$\bar{x}_i$'' (in blue) if and only if $x_i=0$, meaning that for each vertex we include exactly one of the two incoming edges, and exactly one of the two outgoing edges. The resulting graph has $s$ and $t$ connected if and only if $\textsc{parity}(x)=1$.}
    \label{fig:parity}
\end{figure}
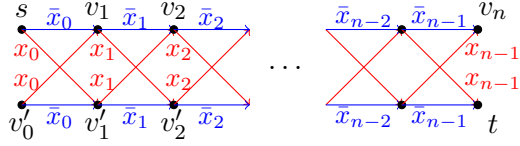
\begin{proof}
We will reduce the \textsc{parity} problem to \textsc{ustcon}\textsubscript{s-arr}. Since \textsc{parity} requires $\Omega(n)$ queries by \lem{parity_complexity}, and the quantum walk oracle $\mathcal{O}_W$ can be implemented using $\widetilde{O}(1)$ sorted array queries by \lem{QW_assumption_implementation}, this reduction will prove the statement of the theorem.

Let $x\in \{0,1\}^n$ be an input of \textsc{parity}. The corresponding output would be $\bigoplus_{i=0}^{n-1}x_i$. Given this, we need to build a \textsc{ustcon}$_{\text{s-arr}}$ input that can be queried with a constant number of queries to $x$. Consider an undirected graph $G=(X,E)$ defined as follows (see also \fig{parity}):
\begin{align*}
X&=\{s=v_0,v_0',v_1,v_1',\ldots,v_{n-1},v_{n-1}',v_n,t=v_n'\}\\
E&=\{\{v_i,v_{i+1}'\},\{v_i',v_{i+1}\}: i\in \{0,\dots,n-1\}, x_{i}=1\}\\
&\qquad \cup \{
\{v_i,v_{i+1}\},\{v_i',v_{i+1}'\} : i\in \{0,\dots,n-1\}, x_{i}=0\}.
\end{align*}
In this setting, $\bigoplus_{i=0}^{n-1}x_i=1$ if and only if $s$ and $t$ are connected in the graph $G$.

Next, we describe how to implement queries ${\cal O}_D$ and ${\cal O}_N$ to $G$ as required by the \textsc{ustcon}$_{\text{s-arr}}$ problem, using queries to ${\cal O}_x$.
Consider the following encoding of vertices of $G$. For $i\in \{0,\dots,n\}$, we let $v_i=(i,0)$, and $v_i'=(i,1)$. That is, for a vertex $(i,b)$, $i\in\{0,\ldots,n\}$ encodes the ``column'' and $b\in\{0,1\}$ encodes the ``row''. Assume that the vertices are ordered lexicographically, i.e.
$$(i,b_i)<(j,b_j)\iff i<j \text{ or } i=j, b_i<b_j.$$

\noindent Queries to $G$ are described according to this ordering.
\begin{enumerate}
    \item Degree queries, ${\cal O}_D$, are trivial in this case as $d_{v_0}=d_{v_0'}=d_{v_n}=d_{v_n'}=1$, and all other degrees are 2.
    \item Since every vertex has degree at most $2$, we explicitly describe neighbour queries, ${\cal O}_N$ for indices $1$ and $2$ such that the ordering assumption holds.
    \begin{itemize}
        \item $\mathcal{O}_N: \ket{i}\ket{b}\ket{1}\ket{0}\ket{0}\mapsto\ket{i}\ket{b}\ket{1}{\color{blue}\ket{i-1}\ket{b}}\overset{\mathcal{O}_x}{\mapsto}\ket{i}\ket{b}\ket{1}{\color{blue}\ket{i-1}\ket{b\oplus x_{i-1}}},\; \forall \; 0<i\le n$\\

        \item $\mathcal{O}_N: \ket{i}\ket{b}\ket{2}\ket{0}\ket{0}\mapsto{\color{blue}\ket{i}}\ket{b}\ket{2}\ket{i+1}{\color{blue}\ket{b}}\overset{\mathcal{O}_x}{\mapsto}{\color{blue}\ket{i}}\ket{b}\ket{2}\ket{i+1}{\color{blue}\ket{b\oplus x_i}},\; \forall \; 0\le i< n$
    \end{itemize}
\end{enumerate}

\noindent It can be seen from the formulas that queries to $G$ can be implemented with a constant number of queries to the parity input $x$.
This implies the $\Omega(n)$ lower bound in $\textsc{ustcon}_{\text{s-arr}}$.

For $\textsc{ustcon}_{\text{qw}}$, note that the graph has bounded degree, and so by \cref{lem:QW_assumption_implementation} we can simulate a query in this model using $O(1)$ queries in the sorted adjacency array model.
This implies a similar $\Omega(n)$ lower bound for this model.
\end{proof}

%-----------------------------------------------------------------------------
\subsection{Metropolis-Hastings walk} \label{sec:MH_description}
%-----------------------------------------------------------------------------

In this section, we consider an unweighted simple graph $G$.
The algorithm that we propose involves a quantum walk on a modified weighted version of $G$ that we call $G'=(X',E',W)$. We start by describing the construction of $G'$ that was introduced in \cite[arXiv~v2]{MH}.

\begin{definition}[Metropolis-Hastings walk]\label{def:MH}
For any graph $G=(X,E)$, the corresponding \emph{Metropolis-Hastings walk} is the random walk on the weighted graph $G'=(X',E',W)$ defined as follows. For every $u\in X$, we include a corresponding vertex $x_u$ in $X'$. In addition, for every edge $\{u,v\}\in E$, we add a new vertex $x_{u,v}$ that splits the edge into two new edges. Formally:
\begin{align*}
    X' &= \{x_u:u\in X\}\cup\{x_{u,v}:\{u,v\}\in E,u<v\}\\
    E' &= \{\{x_u,x_{u,v}\}: \{u,v\}\in E\}.
\end{align*}
For every edge $\{x_u,x_{u,v}\}\in E'$, we define the edge weight $W_{x_u,x_{u,v}}=\frac{1}{d_u}$. These weights define transition probabilities for the random walk on $G'$.
\end{definition}

The above has been called the \emph{cautious walk} in~\cite[arXiv~v2]{MH}, while \emph{Metropolis-Hastings-type} walks are walks in which neighbours are sampled and accepted with some probability.
Our terminology is motivated by the following observation.
If we start with a vertex~$u\in X$ and take two steps of the walk of \cref{def:MH}, then we arrive at another vertex~$v\in X$, which is either the same or a neighbour of~$u$ in~$G$.
The walk on~$G$ defined this way has the following alternative description:
sample a uniformly random neighbour and accept it with probability~$\frac{1/d_v}{1/d_u+1/d_v}=\frac{1}{1+d_v/d_u}$.
This is precisely a random walk that falls into the Metropolis-Hastings framework, justifying our terminology.
The precise choice of acceptance probabilities is sometimes called the \emph{Glauber choice} in the literature (e.g., \cite{lemieux2020efficient}).
We note that a later version of \cite{MH} uses another choice of Metropolis-Hastings walk, but of our purposes we find it convenient to stick to the walk as defined above.

While the hitting time of a random walk between two vertices in $G$ may be as high as $O(n^3)$, in $G'$ it is at most $O(n^2)$~\cite[Lemma 2 of {arXiv} v2]{MH}:
\begin{lemma}[\cite{MH}]\label{lem:MH}
Let $G=(X,E)$ be any unweighted graph, and $G'$ the corresponding (weighted) Metropolis-Hastings graph as in \defin{MH}.
For any $u,v\in X$ connected by a path, ${\cal H}_{u,v}(G')\leq 18n^2$.
\end{lemma}

In order to apply \thm{quantum_walk_search} to $G'$, we need to upper bound ${\sf U}$, the cost of implementing the weighted quantum walk oracle.
For $u\in X'$ the oracle is defined as
$$U:\ket{x}\ket{0}\mapsto\sum_{y\in N(x)}\sqrt{P'_{xy}}\ket{x}\ket{y},\quad \forall x\in X'$$
where $P'_{xy}$ is the probability of walking from $x$ to $y$ defined by the edge weights.

\begin{lemma}\label{lem:Szegedy_implementation}
The weighted quantum walk oracle $U$ for the Metropolis-Hastings walk $G'$ can be implemented with $\widetilde{{O}}(1)$ degree queries $\mathcal{O}_D$, $\widetilde{{O}}(1)$ applications of the quantum walk operator $\mathcal{O}_W$ on the graph $G$, and $\widetilde{O}(1)$ additional gates.

Therefore, $U$ can be implemented with $\widetilde{{O}}(1)$ queries to $G$ in the sorted adjacency array model, and $\widetilde{O}(1)$ additional gates.
\end{lemma}

\begin{proof}
Note that the first statement implies the second one due to \lem{QW_assumption_implementation}. Therefore, we will only prove the first statement.
Consider the following encoding of the vertices of $G'$.
$$X'=\{(u,0): u\in X\}\cup\{(u,v): \{u,v\}\in E,u<v\}\subseteq X\times \big(X\cup\{0\}\big),$$
where $0$ is a null symbol not contained in $X$. The first set of this union corresponds to original vertices of $G$ and the second one corresponds to the added ones.

To implement $U$ on $\ket{x}\ket{0}$, we first compute a bit in an ancilla register $A$ that is $\ket{0}_A$ if $x=(u,0)$ for some $u$, and $\ket{1}_A$ otherwise. We will condition on this value.

First, conditioned on $\ket{0}_A$, our implementation proceeds as follows, for $u\in X$:
\begin{align*}
\ket{x_u}\ket{0}=\ket{u,0}\ket{0} &\overset{J\mathcal{O}_W}{\mapsto} \frac{1}{\sqrt{d_u}}\sum_{v\in N(u)}\ket{u}\ket{v}\ket{u,v}\\
&\overset{J'}{\mapsto} \frac{1}{\sqrt{d_u}}\sum_{v\in N(u)}\ket{u,0}\ket{u,v}=\frac{1}{\sqrt{d_u}}\sum_{v\in N(u)}\ket{x_u}\ket{x_{u,v}}=U\ket{x_u}\ket{0}
\end{align*}

\noindent where $J$ is a unitary that acts as $\ket{u,v}\ket{0}\mapsto \ket{u,v}\ket{u,v}$,
and $J'$ as $\ket{u,v}\ket{u,v}\mapsto \ket{u,0}\ket{u,v}$,
each of which can be implemented with $O(\log (n))$ controlled-NOT gates.

Next, conditioned on $\ket{1}_A$, our implementation proceeds as follows, for $\{u,v\}\in E$ with $u<v$, and $\ket{0}_{A'}$ a fresh ancilla:
\begin{align*}
\ket{0}_{A'}\ket{x_{u,v}}\ket{0}=\ket{0}_{A'}\ket{u,v}\ket{0}  &\overset{1}{\mapsto}\left(\sqrt{\frac{1/d_u}{1/d_u+1/d_v}}\ket{0}+\sqrt{\frac{1/d_v}{1/d_u+1/d_v}}\ket{1}\right)_{A'}\ket{{u,v}}\ket{0}\\
&\overset{2}{\mapsto}\sqrt{\frac{1/d_u}{1/d_u+1/d_v}}\ket{0}_{A'}\ket{{u,v}}\ket{u}+\sqrt{\frac{1/d_v}{1/d_u+1/d_v}}\ket{1}_{A'}\ket{{u,v}}\ket{v}\\
&\overset{3}{\mapsto}\ket{0}_{A'}\left(\sqrt{\frac{1/d_u}{1/d_u+1/d_v}}\ket{u,v}\ket{u}+\sqrt{\frac{1/d_v}{1/d_u+1/d_v}}\ket{{u,v}}\ket{v}\right)\\
&= \ket{0}_{A'}\left(\sqrt{\frac{W_{x_{u,v},x_u}}{w(x_{u,v})}}\ket{x_{u,v}}\ket{x_u}+\sqrt{\frac{W_{x_{u,v},x_v}}{w(x_{u,v})}}\ket{x_{u,v}}\ket{x_v}\right)\\
&=\ket{0}_{A'} U\ket{x_{u,v}}\ket{0},
\end{align*}

\noindent where we use the following mappings:
\begin{description}
\item[$1$:] Query degrees for $u$ and $v$ into a new ancilla register, perform the rotation controlled on the degrees (cf.\ \cite{grover2002creating}), and then uncompute the degrees ($O(1)$ degree queries to $G$).
\item[$2$:] Controlled on the first register, select one of the two vertices to copy into the last register ($O(\log (n))$ Toffoli gates).
\item[$3$:] Flip the bit in $A'$ if the second vertex of $\ket{u,v}$ is the same as the one written in the third register ($O(\log (n))$ elementary gates).
\end{description}
To complete the proof, note that we can uncompute the bit in ancilla $A$, because the register containing $\ket{x}$ has not been changed.
\end{proof}

%-----------------------------------------------------------------------------
\subsection{The algorithm}\label{sec:no-tradeoff}
%-----------------------------------------------------------------------------

We can solve \textsc{ustcon}$_{\mathrm{qw}}(G)$ using \cref{alg:MH}.
This leads to our main theorem of this section.

\begin{algorithm}[ht]
  \caption{Quantum algorithm for \textsc{ustcon}$_{\text{qw}}$ with optimal time and space\label{alg:MH}}
Apply the algorithm from \thm{quantum_walk_search} to the Metropolis-Hastings walk $P'$ with $M=\{t\}$, using \lem{Szegedy_implementation} to implement the quantum walk oracle for $G'$. If the algorithm returns $t$, output ``connected'', and otherwise output ``disconnected''.
\end{algorithm}

\begin{theorem}\label{thm:single-walk}
There exists a $O(\log (n))$-space quantum algorithm that decides \textsc{ustcon}$_{\mathrm{qw}}$ and \textsc{ustcon}$_{\mathrm{s-arr}}$ with bounded one-sided error in $\widetilde{O}(n)$ gates and queries.
\end{theorem}
\begin{proof}

Let $X_s \subseteq X$ denote the connected component of $s$.
If $t\in X_s$, the algorithm will output~$t$ with probability at least $2/3$, in which case our algorithm will output the correct answer, ``connected''. If $t\not\in X_s$, then the algorithm will output an element of $X_s$ with probability 1, in which case, our algorithm will output the correct answer ``disconnected''. This establishes correctness of \cref{alg:MH} with one-sided error.

To analyze the complexity, note that by \lem{Szegedy_implementation} we have ${\sf U} = \widetilde{O}(1).$
For any $u\in X_s$, we can check if $u\in M$ by checking if $u=t$ in complexity ${\sf C}=O(\log (n))=\widetilde{O}(1)$. To complete the analysis, we need only upper bound the commute time between $s$ and $t$ when $t \in X_s$.
Since $\mathcal{C}_{s,t}=\mathcal{H}_{s,t}+\mathcal{H}_{t,s}$, by \lem{MH}, we have $\mathcal{C}_{s,t}\leq 36n^2=:{\cal C}$. Thus, referring to \thm{quantum_walk_search}, the complexity of our algorithm is:
\begin{equation*}\widetilde{O}\left(\sqrt{{\cal C}\log({\cal C})\log(\log({\cal C}))}({\sf C}+{\sf U})\right) = \widetilde{O}(n).\qedhere
\end{equation*}
\end{proof}

%=============================================================================
\section{Time-space tradeoff for bounded spectral gap}\label{sec:tradeoff}
%=============================================================================
In this section we revisit the problem of undirected $st$-connectivity in the setting where one is given a lower bound on the spectral gap of the random walk.
As discussed in \cref{sec:random-walks}, such a bound is tightly related to the mixing time of the walk.
We will give a quantum algorithm that exhibits a nontrivial time-space tradeoff in this setting.

Our discussion will be general and apply to random walks on weighted graphs as defined in \cref{eq:transmat}.
This is useful since the spectral gaps and mixing times of random walks on $G$ with different edge weights are in general \emph{not} comparable.
E.g., on the lollipop graph (an $n$-vertex clique connected to an $O(n)$-vertex path) the mixing time of the unweighted random walk is $\Theta(n^3)$ \cite{brightwell1990maximum}, while it is $O(n^2)$ for the Metropolis-Hastings walk.\footnote{This follows from the $O(n^2)$ upper bound on the maximum hitting time of the Metropolis-Hastings walk (\cref{lem:MH}), and the fact that the maximum hitting time upper bounds the mixing time \cite[Lemma 10.2]{levin2017markov}.}
On the other hand, on an $n$-vertex star graph the unweighted random walk has mixing time $O(1)$ while the Metropolis-Hastings walk has mixing time $\Theta(n)$.
Thus, while the specific edge weights do not affect whether $s$ and $t$ are connected, they do impact the algorithm. Throughout this section, we assume some fixed edge weights are given, and we do not try to optimize for ``good'' edge weights.
More specifically, we assume access to a \emph{weighted quantum walk oracle} that for every vertex outputs a superposition of its neighbours, with squared amplitudes proportional to the edge weights:
\begin{align*}
    \mathcal O_W \colon \ket u \ket 0 \mapsto \sum_{v\in N(u)} \sqrt{\frac{W_{u,v}}{w_u}} \ket u \ket v
    = \sum_{v\in N(u)} \sqrt{P_{u,v}} \ket u \ket v
    \qquad \forall u \in X
\end{align*}

Moreover, we assume access to the weighted vertex degrees $w_u$ and that these degrees are of bounded absolute value for some poly$(n)$ bound. This will allow us to generate the state $\ket{\pi_{X'}}$ for any subset $X' \subseteq X$ stored in QCRAM.

\begin{problem}[$\textsc{ustcon}_{\text{qw},\delta}$]\label{prob:ustcon-gap}
Given access to an undirected weighted graph via the quantum walk oracle $\mathcal{O}_W$, two vertices $s,t\in X$, and the promise that either $s$ and $t$ are disconnected or the spectral gap of the transition matrix of the walk is at least some~$\delta>0$, decide which is the case.
\end{problem}

\noindent Our main result of this section is the following:

\begin{restatable}{theorem}{maingap}
\label{thm:trade-off}
Fix $\delta\geq 0$. Let ${\cal G}_n$ be a family of undirected weighted graphs $G=(X,E,W)$ with $n=|X|$, such that $\gamma_\star(G)\geq \delta$ whenever $s$ and $t$ are connected.
Then for any $S=\Omega(\log (n))$, there is a quantum algorithm that decides \textsc{ustcon}$_{\mathrm{qw},\delta}$ on ${\cal G}_n$ with bounded error in $O(S)$ space -- of which $O(\log(n/\delta))$ is quantum memory, and the remainder is QCRAM -- and $T=\widetilde{O}(\frac{S}{\delta}\log(\frac{1}{\pi_{\min}})+\sqrt{\frac{n}{\delta S}})$ queries to ${\cal O}_W$, elementary gates, and QCRAM queries.
\end{restatable}

Note that we can assume $\delta \geq 1/n$. If $\delta < 1/n$ then $T \approx S/\delta$. There is no time-space tradeoff, and it is always faster to run the Metropolis-Hastings algorithm (\cref{alg:MH}).

The algorithm is stated below as \cref{alg:tradeoff}.
It consists of three stages.
We fix some parameter $p$, which denotes the number of ``pebbles'', or vertices the algorithm will keep track of (so $S=O(p\log(n))$).
First, we run $O(p)$ classical random walks starting from $s$, each of length $\ell=O\left(\frac{1}{\delta}\log(\frac{n}{p\pi_{\min}})\right)$.
This allows us to sample a set $L$ of $O(p)$ points from $X_s$, the connected component of $s$ (the big-O notation suppresses a universal constant that is given in the proof). Since $\ell$ is at least the mixing time of $G$ (see \thm{mix}), assuming $s$ and $t$ are connected, each point is sampled (approximately) from $\pi$. We do the same from $t$ to get a random subset $M \subseteq X_t$ connected to $t$.

Next, we use $L$ and $M$ to prepare (up to some error) the states $\ket{\pi_{X_s}}$ and $\ket{\pi_{X_t}}$, using \emph{inverse quantum walk search}, which we describe in more detail in \sec{ST-tradeoff-state-prep}. If $s$ and $t$ are in the same connected component, then $\ket{\pi_{X_s}}=\ket{\pi_{X_t}}$, and otherwise, the states are orthogonal. The final step is to distinguish these two cases using a SWAP test.
This roughly follows an earlier approach in \cite{apers2019qsampling}, the main difference being that we sample the sets $L$ and $M$ using a random walk (which allows us to exploit the gap promise), while in \cite{apers2019qsampling} the sets are constructed using a breadth-first search.

\begin{algorithm}[ht]
  \caption{Quantum algorithm for \textsc{ustcon}$_{\text{qw}}$ with a tradeoff\label{alg:tradeoff}}
  \texttt{Seed set:}\quad Run $O(p)$ classical random walks from $s$ and $O(p)$ classical random walks from~$t$, each for $O(\frac1{\delta} \log (\frac n {p \pi_{\min}}))$ steps.
  Let $L$ and $M$ denote the respective sets of endpoints, without duplicates.
  If $L \cap M \neq \emptyset$, return ``connected''.

  \texttt{State preparation:}\quad Run inverse quantum walk search from $\ket{\pi_L}$ and $\ket{\pi_M}$ for time $\widetilde{{O}}\left(\sqrt{\frac{n}{\delta p}}\right)$ to prepare $\ket{\pi_{X_s}}$ and $\ket{\pi_{X_t}}$, respectively, to precision $1/8$.

  \texttt{SWAP test:}\quad Do a SWAP test on the resulting states.
  If the test returns ``0'', return ``connected'', otherwise return ``disconnected''.
\end{algorithm}

If we specialize \cref{thm:trade-off} to the unweighted graph case, we get the following corollary.
\begin{corollary}
For any $S > 0$, there is a quantum algorithm that solves $\textsc{ustcon}_{\mathrm{s\mbox{-}arr}}$ with a promise $\delta > 0$ on the random walk spectral gap using space $S$ and time $T \in \widetilde O(S/\delta + \sqrt{n/(\delta S)})$.
\end{corollary}

An analogous result holds for the Metropolis-Hastings walk described in \cref{sec:MH_description} given a promise on its spectral gap, since we showed that the corresponding quantum walk oracle can also be efficiently implemented.

In the remainder of this section we will analyze each stage of \cref{alg:tradeoff}.

\subsection{Analysis of step 1: Seed set}\label{sec:step1_analysis}

Recall that the first stage results in random sets~$L = \{ x_1,\dots,x_{c p} \}$ and $M = \{ y_1,\dots,y_{c p} \}$, where $x_1,\dots,y_{cp}$ are the endpoints of independent random walks starting at $s$ or $t$, respectively, and $c>0$ is some universal constant that we will choose later.
Since we run those random walks for $O(\frac1{\delta} \log (\frac n {p \pi_{\min}}))$ steps, by \cref{thm:mix} it follows that the $x_j$ are independent samples drawn from a distribution $\tilde\pi$ such that
\begin{align*}
  \norm{\tilde\pi - \pi}_{\TV} \leq \frac p{8n},
\end{align*}
where $\pi$ is the stationary distribution on $X_s$.
If $s$ and $t$ are connected then $X_s = X_t$ and the samples $y_j$ are similarly drawn from a distribution that is $p/(8n)$-close to $\pi$.
Here we prove that this implies lower bounds on the stationary measure of the sets $L$ and $M$.

\begin{proposition}\label{prp:M_measure}
There exists a universal constant $c>0$ such that the following holds.
Let~$p \in [n]$ and assume $L\subseteq X$ is a random set obtained by sampling $c p$ independent elements from a distribution $\tilde{\pi}$ such that $\norm{\tilde\pi - \pi}_{\TV} \leq \frac p {8n}$ (and removing duplicates).
Then, $\Pr(\pi(L) \geq \frac p{8n}) \geq \frac9{10}$.
\end{proposition}

The proof of the proposition uses the following lemma, which formalizes the intuition that adding a random element to a low-probability subset should increase the probability.

\begin{lemma}\label{clm:set_growth_with_a_new_sample}
Let $X$ be a set of cardinality~$n$, $A \subseteq X$ an arbitrary fixed subset, and let $b$ be drawn at random from an arbitrary distribution $\sigma$ on~$X$.
Then:
\[ \Pr\mleft( \sigma(A \cup \{b\}) > \sigma(A) + \frac {1 - \sigma(A)} {2n}\mright) \geq \frac {1 - \sigma(A)} 2. \]
\end{lemma}
\begin{proof}
Say $b$ is \emph{bad} if $b \in A$ or $\sigma(b) \leq \frac {1 - \sigma(A)} {2n}$, and \emph{good} otherwise.
Then
\begin{align*}
  \Pr\mleft( \sigma(A \cup \{b\}) > \sigma(A) + \frac {1 - \sigma(A)} {2n}\mright) = \Pr(\text{$b$ is good}) = 1 - \Pr(\text{$b$ is bad}).
\end{align*}
We can compute:
\begin{align*}
  \Pr(\text{$b$ is bad})
= \sigma\mleft( A \cup \left\{ x \in X : \sigma(x) \leq \frac {1 - \sigma(A)} {2n}\right\} \mright)
\leq \sigma(A) + n \cdot \frac {1 - \sigma(A)} {2n}
= \frac {1 + \sigma(A)} 2,
\end{align*}
from which the claim follows.
\end{proof}

\begin{proof}[Proof of \cref{prp:M_measure}]

Let $x_1, x_2, \dots$ denote samples drawn independently at random from $\tilde\pi$.
For any integer $T\geq1$, define $L_T := \{ x_1, \dots, x_T \}$ as the set consisting of the first $T$ samples (with duplicates removed), as well as $L_0 := \emptyset$.
We say that the $T$-th sample is a \emph{success} if
\begin{align*}
  \tilde\pi(L_{T-1}) \geq \frac12 \quad\text{or}\quad \tilde\pi(L_T) \geq \tilde\pi(L_{T-1}) + \frac1{4n},
\end{align*}
and a \emph{failure} otherwise.
Let $T_j$ denote the index of the~$j$-th success, with $T_0 := 0$.
Then, clearly,
\begin{align*}
  \tilde\pi(L_{T_p}) \geq \min \left\{ \frac12, \frac p {4n} \right\} = \frac p {4n}
\quad\text{and hence}\quad
  \pi(L_{T_p}) \geq \tilde\pi(L_{T_p}) - \frac{p}{8n} \geq \frac p {8n}.
\end{align*}

On the other hand, note that by \cref{clm:set_growth_with_a_new_sample}, the $T$-th sample is a success with probability at least $\frac14$, even if we condition on all prior samples.
In particular, the probability that there are~$k$ failures in a row is at most $(\frac34)^k$.
Therefore,
\begin{align*}
  \E[T_j - T_{j-1}]
= \sum_{k=1}^\infty k\Pr(T_j = T_{j-1}+k)
\leq \sum_{k=1}^\infty k\left( \frac34 \right)^k
= 12
\end{align*}

and hence, by Markov's inequality,
\begin{align*}
  \Pr(T_p > c p)
\leq \frac 1 {cp} \E[T_p]
= \frac 1 {cp} \sum_{j=1}^p \E[T_j - T_{j-1}]
\leq \frac {12p} {cp} = \frac1{10}
\end{align*}
provided we choose $c := 120$.
Since the random set $L$ in the statement of the lemma is defined by taking $cp$ many samples, we obtain that
\begin{equation*}
  \Pr\mleft( \pi(L) \geq \frac p {8n} \mright)
\geq \Pr(T_p \leq c p)
\geq 1 - \frac1{10} = \frac9{10}.
\qedhere
\end{equation*}
\end{proof}

\subsection{Analysis of step 2: State preparation}\label{sec:ST-tradeoff-state-prep}
Now we turn to the analysis of the quantum walk search routine in step~2 of the algorithm.
We rely on the following proposition from \cite{apers2019qsampling}, which formalizes the idea of ``inverse quantum walk search''.\footnote{\cite[Proposition 1]{apers2019qsampling} only proves the proposition for simple random walks, however it trivially extends to random walks on weighted graphs.
}
\begin{proposition}[{\cite[Proposition 1]{apers2019qsampling}}]
\label{prop:inv-qw-search}
Consider a subset $A \subseteq X$ of a (connected) graph $G$, and let $\delta$ be a lower bound on the spectral gap of a random walk $P$ on $G$ with stationary distribution $\pi$.
From $\ket{\pi_A}$, we can generate a state $\ket{\tilde\pi} = \ket{\pi} + \ket{\Gamma}$ with $\| \ket{\Gamma} \|_2 \leq \epsilon$ using an expected number of calls to the weighted quantum walk oracle
\[
O\left( \frac{1}{\sqrt{\pi(A) \delta}}
    \log\left( \frac{1}{\pi(A) \epsilon} \right) \right),
\]
and $O(1/\sqrt{\pi(A) \delta})$ reflections around $\ket{\pi_A}$. The algorithm uses space logarithmic in $n$, $1/\delta$, $1/\pi(A)$ and $1/\epsilon$.
\end{proposition}

The proposition implies the following.

\begin{claim}\label{clm:step2_tradeoff}
Step 2 of \cref{alg:tradeoff} prepares $1/8$-approximations of $\ket{\pi_{X_s}}$ and $\ket{\pi_{X_t}}$ with probability at least $9/10$ in time complexity $O\left( \sqrt{\frac{n}{p\delta}} \log\left(\frac{n}{p}\right) \right)$.
\end{claim}
\begin{proof}
First, step $2$ prepares superpositions $\ket{\pi_L}$ and $\ket{\pi_M}$.
Applying \lem{QCRAM} with $l=n$ (upper bound on the set size) and $k=\log(n)$ (number of bits necessary to write down a vertex index) and the assumption that weighted vertex degrees $w_u$ are of bounded absolute value for some poly$(n)$ bound, shows that this superpositions preparation can be done in $\widetilde{O}(1)$ time.
By \cref{prp:M_measure}, we have that $\pi(L),\pi(M) \geq p/(8n)$ with probability at least $9/10$.
By \cref{prop:inv-qw-search} and our preceding remark, we can then prepare $1/8$-approximations of $\ket{\pi_{X_s}}$ and $\ket{\pi_{X_t}}$ in time
\begin{equation*}
O\left( \sqrt{\frac{n}{p\delta}} \log\left(\frac{n}{p}\right) \right).\qedhere
\end{equation*}
\end{proof}

\subsection{Analysis of step 3: SWAP test}\label{sec:step3_analysis}
In the last step of our algorithm, we wish to decide whether $\ket{\pi_{X_s}} = \ket{\pi_{X_t}}$ or whether they are orthogonal.
For this we use the SWAP test.
\begin{claim}\label{clm:step3_tradeoff}
Step 3 of \cref{alg:tradeoff} decides whether $\ket{\pi_{X_s}} = \ket{\pi_{X_t}}$ or whether they are orthogonal with constant probability in time $\widetilde{{O}}(1)$.
\end{claim}

\begin{proof}

Using a single copy of two states $\ket{\psi}$ and $\ket{\psi'}$, and $O(\log (n))$ additional gates, the SWAP test returns ``0'' with probability $(1+|\langle \psi|\psi' \rangle|^2)/2$ and ``1'' with probability $(1-|\langle \psi|\psi'\rangle|^2)/2$ \cite[Claim 3.1]{aharonov2003adiabatic}.

By \cref{clm:step2_tradeoff}, in step 2 we prepared states $\ket{\tilde\pi_{X_s}} = \ket{\pi_{X_s}} + \ket{\Gamma_L}$ and $\ket{\tilde\pi_{X_t}} = \ket{\pi_{X_t}} + \ket{\Gamma_M}$ such that $\|\ket{\Gamma_L}\|_2,\|\ket{\Gamma_M}\|_2 \leq 1/8$ with probability $9/10$.
By a triangle inequality, this implies that
\[
\Big| |\langle \tilde\pi_{X_s} | \tilde\pi_{X_t} \rangle |
- | \langle \pi_{X_s} | \pi_{X_t} \rangle | \Big|
< 1/3,
\]
and so $|\langle \tilde\pi_{X_s} | \tilde\pi_{X_t} \rangle | > 2/3$ if $s$ and $t$ are connected, but $|\langle \tilde\pi_{X_s} | \tilde\pi_{X_t} \rangle | < 1/3$ otherwise.
The SWAP test distinguishes these cases with constant probability.
\end{proof}

\subsection{Proof of Theorem~\ref{thm:trade-off}}\label{sec:trade-off-proof}

In this section, we prove \thm{trade-off}, which we restate here for convenience.

\maingap*
\begin{proof}
We first analyze the space complexity of \cref{alg:tradeoff}. Step 1 is purely classical, and uses $O(p\log (n))$ space to store the $O(p)$ vertices in $L$ and $M$, with each random walk using $O(\log (n))$ space. We can implement step 2 using the algorithm referred to in \cref{prop:inv-qw-search} using $O(\log (n))$ qubits of space, but this requires that the $O(p\log (n))$ classical space used to store $L$ and $M$ in step 1 is QCRAM. Finally, step 3 just uses $O(\log (n))$ quantum space. Thus, the claimed space complexity follows if we set $S=p\log (n)$.

Next, we analyze the time complexity. Every random walk of step 1 adds $O\left(\frac{1}{\delta}\log(\frac{n}{p\pi_{\min}})\right)$ to the time complexity. The total number of walks is $O(p)$. Checking whether $L\cap M=\emptyset$ is $O(p)$ as this is the total number of samples. Hence, the overall complexity of the first step is $\widetilde{O}\left(\frac{p}{\delta}\log(\frac{1}{\pi_{\min}})\right)$.
By \clm{step2_tradeoff}, the complexity of step 2 is $O\left( \sqrt{\frac{n}{p\delta}} \log\left(\frac{n}{p}\right) \right)$. Finally, the SWAP test in step 3 uses only $O(\log (n))$ gates, since the states being compared are $O(\log (n))$-qubit states. Hence, the total time complexity of \cref{alg:tradeoff} is $$T=\widetilde{O}\left(\frac{p}{\delta}\log\left(\frac{1}{\pi_{\min}}\right)+\sqrt{\frac{n}{\delta p}}\right)=\widetilde{O}\left(\frac{S}{\delta}\log\left(\frac{1}{\pi_{\min}}\right)+\sqrt{\frac{n}{\delta S}}\right),$$
since $S=\widetilde{O}(p)$.

Finally, for the correctness of the algorithm, by \clm{step3_tradeoff}, \cref{alg:tradeoff} distinguishes between the case where $\ket{\pi_{X_s}}$ and $\ket{\pi_{X_t}}$ are equal and the case where they are orthogonal with bounded error. If $X_s=X_t$ (i.e.~$s$ and $t$ are connected) then the states are equal, and if $X_s\cap X_t=\emptyset$ (i.e.~$s$ and $t$ are not connected) then they are orthogonal.
\end{proof}

\bibliographystyle{alpha}
\bibliography{refs}

\end{document}